\documentclass[journal,twocolumn,10pt]{IEEEtran}
\usepackage[T1]{fontenc}
\usepackage{cite}
\usepackage[cmex10]{amsmath}
\usepackage{amsthm}
\newtheorem{remark}{Remark}
\usepackage{amssymb}
\usepackage{bbm}
\usepackage{mathtools}
\usepackage[caption=false,font=normalsize,labelfont=sf,textfont=sf]{subfig}
\usepackage[hidelinks]{hyperref}
\usepackage{algorithmic}
\usepackage{textcomp}
\usepackage{cite}
\usepackage[usenames, dvipsnames]{color}
\usepackage{subcaption}
\usepackage{graphicx}
\usepackage{float}
\usepackage{url}
\usepackage[usestackEOL]{stackengine}
\usepackage{verbatim}
\usepackage{dsfont}
\usepackage{balance}

\theoremstyle{definition}

\newtheorem{lemma}{Lemma}

\newcommand{\mbf}[1]{\mathbf{#1}}
\newcommand{\mbb}[1]{\mathbb{#1}}

\newcommand{\mymod}[1]{\mathrm{mod}}

\usepackage{graphicx}

\begin{document}

\title{On the Performance of Randomly Generated Codebooks in Near-Field Communications}

\author{{\IEEEauthorblockN{Akarsh Yadav,~\IEEEmembership{Graduate Student Member, IEEE}, Ly V. Nguyen,~\IEEEmembership{Member, IEEE}, James V. Krogmeier,~\IEEEmembership{Senior Member, IEEE}, and David J. Love,~\IEEEmembership{Fellow, IEEE} }}
\thanks{Akarsh Yadav, James V. Krogmeier, and David J. Love are with the Elmore Family School of Electrical and Computer Engineering, Purdue University, West Lafayette, IN 47907 USA (e-mails: {yadav111, jvk, djlove}@purdue.edu).}
\thanks{Ly V. Nguyen was with the Elmore Family School of Electrical and Computer Engineering, Purdue University, West Lafayette, IN 47907 USA. He is now with the Department of Electrical Engineering and Computer Science, University of Kansas, KS 66045 USA (e-mail: vanly.nguyen@ku.edu).}
\thanks{This work was supported in part by National Science Foundation under Grant EEC-1941529, Grant CNS-2212565, and Grant CNS-2225578.
}
}
\maketitle
\begin{abstract}
Near-field communication with large antenna arrays requires beamforming techniques that account for spherical wavefront effects. While predefined codebooks are widely used in multiple-antenna wireless systems, the performance of randomly generated codebooks in the near-field region has not been investigated. In this work, we study four randomly generated codebooks for near-field communications: random vector quantization (RVQ), random phase (RP), random far-field (RFF), and random near-field (RNF) codebooks. We use the known RVQ gain distribution as a benchmark. For RP, we derive a large-array bulk approximation and characterize the high-gain tail behavior of the gain distribution. For RNF, we characterize the high-gain tail behavior of the gain distribution. For RFF, we derive an analytical approximation for the maximum achievable beamforming gain when far-field beams are used to serve near-field users, along with the corresponding maximizing far-field direction. We also analyze the maximum-gain scaling with codebook size. Monte Carlo simulations verify the analytical results and show that RVQ and RP degrade significantly as the array size increases, while RNF achieves the most favorable high-gain scaling. RFF can outperform RVQ and RP for finite codebook sizes, but its maximum gain saturates below one due to the phase mismatch between the far-field beams and the near-field channel. These results show how codebook structure, array size, and codebook size affect random beamforming performance in the near-field.
\end{abstract}

\begin{IEEEkeywords}
Near-field, random codebooks, beamforming, codebook design, random vector quantization, random phase beamforming, far-field beamforming.
\end{IEEEkeywords}

\IEEEpeerreviewmaketitle

\section{Introduction}
Future wireless communication systems aim to support extremely high data rates through large antenna arrays and increased bandwidth \cite{10908559}. Large antenna arrays provide high beamforming gain, improved spatial resolution, and highly directional multi-user transmission \cite{8732419,10365224}. However, these benefits come with challenges due to the need for highly directional transmission using limited feedback-rates, sometimes in severe propagation conditions. Moreover, the small wavelengths in future communication systems allow for more antennas in a given aperture \cite{10558818}.

As the physical aperture increases, the propagation behavior changes. In particular, the Fraunhofer distance scales quadratically with the array aperture, enlarging the near-field region \cite{7942128}. Consequently, future wireless systems with very large arrays are expected to serve many users in the near-field region \cite{10496996}. In this region, the planar-wave model is no longer accurate. Instead, the channel must be modeled using a spherical wavefront that depends on both the user's angle and distance \cite{9903389,10098681}. Therefore, beamforming and codebook design must account for both angular and distance variations.

This near-field behavior changes the structure of beamforming codebooks. In conventional far-field systems, codebooks usually cover only the angular domain. For example, the discrete Fourier transform (DFT) codebook in a uniform linear array (ULA) provides far-field beams that uniformly sample the angular domain \cite{6600706}. Codebook-based beamforming is also used in practical standards. In 5G NR, CSI feedback relies on standardized Type-I and Type-II codebooks, where Type-I codebooks provide lower-overhead beam selection and Type-II codebooks provide higher-resolution channel representation at the cost of increased feedback overhead \cite{11346973}. Similarly, millimeter-wave (mmWave) Wi-Fi systems such as IEEE 802.15.3c implement a codebook-based beamforming protocol \cite{5936164}. However, in the near-field environment, the channel contains an additional distance-dependent phase term, causing phase mismatch when far-field codebooks are employed. This has motivated recent work on polar-domain and hierarchical near-field codebooks that jointly search over angle and distance \cite{9693928,6600706,11183591,10195974,10506674,10365224,10304223}.

Despite this progress, the performance of random codebooks in the near-field regime remains completely unexplored. In previous codebook designs, random codebooks have been used to study codebook performance across different quantization rates because they are simple to construct, analytically tractable, and useful as benchmarks for structured codebooks \cite{4100151,7842197,7279196}. However, near-field channels contain distance-dependent phase curvature, which changes how beams concentrate energy. Therefore, the impact of near-field propagation on the performance of random codebooks remains largely unknown.

In this work, we investigate the performance of random codebooks for near-field beamforming. We evaluate four classes of random codebooks: random vector quantization (RVQ), random phase (RP), random far-field (RFF), and random near-field (RNF) codebooks. RVQ draws beamforming vectors independently from the uniform distribution on the complex unit sphere and is widely used as a benchmark in limited-feedback beamforming \cite{4100151}. RP uses constant-modulus beamforming vectors with random antenna phases and has been studied for mmWave initial access and channel-parameter estimation \cite{7842197}. RFF consists of randomly selected far-field steering vectors and is motivated by opportunistic random beamforming in mmWave systems \cite{7279196}. Unlike RVQ, RP, and RFF, there is no prior work on RNF codebooks. We introduce RNF as a near-field random benchmark, where each codeword is generated by randomly sampling an angle-distance pair and forming the corresponding near-field steering vector. These codebooks represent different levels of channel-structure information: RVQ randomizes both antenna amplitudes and phases; RP maintains constant antenna amplitudes while randomizing the phases; RFF imposes a linear phase structure of far-field steering vectors; and RNF samples directly from the near-field steering vector set.

The main contributions of this work are summarized as follows:
\begin{itemize}
    \item We characterize the beamforming-gain distributions of RVQ, RP, RFF, and RNF codebooks in near-field channels and compare their limiting behavior as the codebook size increases.
    For the RP codebook, we derive two approximations for the beamforming-gain distribution. The first is a large-array bulk approximation based on the central limit theorem, while the second describes the probability of obtaining a beamforming gain close to one.
    \item For the RFF codebook, we derive an analytical approximation for the maximum gain achievable by far-field beams when serving a near-field user, along with the corresponding maximizing far-field beam direction. We also show that, even as the codebook size increases, RFF is limited by the residual phase mismatch between far-field beams and the near-field channel.
    \item For the RNF codebook, we characterize the probability of obtaining a beamforming gain close to one. We show that this probability decreases only linearly as the target gain approaches one.
    \item We compare the asymptotic maximum-gain scaling of the codebooks and show that RNF has the most favorable scaling, followed by RP and RVQ, while RFF is limited due to phase mismatch.
\end{itemize}
\textit{Notation:} Scalars are denoted by lower-case letters, vectors by bold lower-case letters, and matrices by bold upper-case letters. Sets are denoted by calligraphic letters. The superscripts $(\cdot)^{\mathsf T}$ and $(\cdot)^{\mathsf H}$ denote the transpose and the conjugate transpose, respectively. The Euclidean norm and squared Euclidean norm of a vector are denoted by $\lVert\cdot\rVert_2$ and $\lVert\cdot\rVert_2^2$, respectively, while $|\cdot|$ denotes the absolute value of a scalar. The operators $\mathbb E[\cdot]$ and $\Pr{\cdot}$ denote expectation and probability, respectively. The notation $\mathcal{CN}(\mathbf 0,\mathbf I_M)$ denotes a circularly symmetric complex Gaussian random vector with zero mean and covariance matrix $\mathbf I_M$, and $\mathrm{Unif}(\cdot)$ denotes the uniform distribution. The sets of integers and nonnegative integers are denoted by $\mathbb Z$ and $\mathbb N_0$, respectively. The notation $x\mapsto y$ is used to denote an index transformation or change of variable.
\section{System Model}\label{system_model_section}
\subsection{System Model}
We consider a single cell downlink system consisting of a base station (BS) and a single-antenna user. The BS is equipped with a ULA with $M$ antennas. Let $\mathbf{w} \in \mathbb{C}^{M}$ denote the transmit beamforming vector, with $\lVert\mathbf{w}\rVert_{2} = 1$. Then, the received signal at the user can be written as
\begin{align}\label{system_model}
    y = \mathbf{h}^{\mathsf{H}}\mathbf{w}s + \eta,
\end{align}
where $\mathbf{h} \in \mathbb{C}^M$ is the downlink channel vector, $s$ is the transmitted symbol with an average power of $P$, i.e., $E[\lvert s \rvert^2] = P$, and $\eta \sim \mathcal{CN}(0,\sigma^2)$ is additive white Gaussian noise.

Since the planar-wave approximation is inaccurate for users close to large antenna arrays, we use a spherical-wave propagation model \cite{10220205}. Accordingly, we consider a near-field scenario and model the channel vector as
\begin{align}
    \mathbf{h}(\theta,r) = \alpha \mathbf{a}(\theta,r),
\end{align}
where $\alpha$ is the complex path gain and $\mathbf{a}(\theta,r)$ is the near-field steering vector for a user located at a normalized direction $\theta$ and distance $r$.

Let $\lambda$ denote the carrier wavelength and $d = \frac{\lambda}{2}$ denote the antenna spacing. The ULA is placed along the X-axis, with the first antenna at the origin and the $n$-th antenna at $x_n = nd$ for $n =0,1,\ldots,M-1$. The user distance is measured from the first antenna element. If the user is located at a distance $r$ and in a normalized direction $\theta$ with respect to the array broadside\footnote{If $\vartheta$ denotes the physical angle measured from the array broadside, then $\theta = \sin(\vartheta)$. Thus, $\theta =0$ corresponds to broadside, and $\theta = \pm 1$ corresponds to directions parallel to the array axis.}, then the distance from the user to the $n$th antenna is \cite{goodman1969introduction,balanis2016antenna}
\begin{align}
    r^{(n)}(\theta,r) &= \sqrt{r^2 + n^2d^2 - 2r\theta nd} \\
    &= r - nd\theta + \frac{n^2d^2}{2}\frac{1-\theta^2}{r} + \mathcal{O}\left(\frac{1}{r^2}\right), \label{r_approx}
\end{align}
for $n =0,1,\ldots,M-1$. Throughout the paper, we write $r^{(n)}(\theta,r)$ as $r^{(n)}$ when the dependence on $(\theta,r)$ is clear. The approximation in \eqref{r_approx} follows from a second-order expansion of the exact distance expression and highlights the two main components of the near-field phase profile: a linear term in $n$, which is also present in the far-field regime, and a quadratic term in $n$, which is unique to near-field propagation.

Using this distance expression, the corresponding near-field steering vector is given by
\begin{align}
    \mathbf{a}^{}(\theta_{},r_{}) = \frac{1}{\sqrt{M}} \left[1,\,e^{-j \frac{2\pi}{\lambda}\left(r^{(1)} - r\right) },\,\ldots,\,e^{-j \frac{2\pi}{\lambda}\left(r^{(M-1)}-r\right)} \right]^{\mathsf{T}}.
\end{align}
The factor $\frac{1}{\sqrt{M}}$ ensures that $\lVert \mathbf{a}(\theta,r) \rVert_{2}^2 =1$, independent of the array size. The higher order term in \eqref{r_approx} can be neglected when the user is in the radiating near-field, i.e., $r > r_{\min}$, where $r_{\min} = 0.62 \sqrt{\frac{D^3}{\lambda}}$\cite{balanis2016antenna}. Here, $D$ represents the array aperture. For large arrays, we use $D \triangleq Md \approx (M-1)d \approx (M+1)d$ when deriving the large-array approximations \cite{10002944}.

Let the actual user location be $(\theta_t,r_t)$, where $\theta_t$ and $r_t$ represent the user's normalized direction and distance, respectively. Then, we model the corresponding channel vector as 
\begin{align}
    \mathbf{h}(\theta_t,r_t) = \frac{\alpha}{\sqrt{M}}\left[ e^{-j \varphi_{0}(\theta_t,r_t)},\,\ldots\,,\,e^{-j \varphi_{M-1}(\theta_t,r_t) }\right]^{\mathsf{T}},    
\end{align}
where the phase of the $n$-th entry of the array response is given by
\begin{align}
    \varphi_{n}(\theta_t,r_t) = \frac{2\pi}{\lambda}\left(-nd \theta_t + \frac{n^2 d^2}{2}\frac{1-\theta_t^2}{r_t} \right).
\end{align}

We assume that the user location is random but confined to a known region $\theta_t \in [\theta_{\min},\theta_{\max}]$ and $r_t \in [r_{\min},r_{\max}]$. We assume that these bounds are known to the transmitter. Within this region, the normalized direction $\theta_t$ is uniformly distributed, and its probability density function (PDF) is given by 
\begin{align}
    f_{\Theta}(\theta_t) &= \frac{1}{\theta_{\text{max}} - \theta_{\text{min}}}, \qquad \theta_{\min} \leq \theta_t \leq \theta_{\max},
\end{align}
and zero otherwise.
The user distance $r_t$ follows a linear distribution over $[r_{\min}, r_{\max}]$, corresponding to a uniform spatial distribution across the considered polar region. Therefore, the PDF of $r_t$ is given by
\begin{align}
    f_{R}(r_t) &= \frac{2r_t}{ r_{\max}^2 - r_{\min}^2}, \qquad r_{\min} \leq r_t \leq r_{\max},
\end{align}
and zero otherwise.
\subsection{Codebook Design}\label{section_codebook_design}
We consider a codebook $\mathcal{F}$ containing $K$ codewords, where each codeword is generated randomly according to a prescribed distribution. The codebook is written as $\mathcal{F} = \{ \mathbf{w}_{1}, \mathbf{w}_{2}, \cdots , \mathbf{w}_{K} \}$,
where each codeword satisfies the unit-norm constraint. The codewords are assumed to be generated independently from a common distribution, so the statistical behavior of the codebook is determined by the distribution of a single randomly-generated beamforming vector.

For a user located at $(\theta_t,r_t)$, the channel is $\mbf{h}(\theta_t,r_t) = \alpha \mbf{a}(\theta_t,r_t)$, where $\mbf{a}(\theta_t,r_t)$ is the unit-norm near-field array response. The received power associated with codeword $\mathbf{w}_{k}$ is proportional to
\begin{align}
    \lvert\mathbf{h}(\theta_t,r_t)^{\mathsf{H}}\mathbf{w}_k\rvert^2
    = \lvert\alpha\rvert^2 \lvert\mathbf{a}(\theta_t,r_t)^{\mathsf{H}}\mathbf{w}_k\rvert^2.
\end{align}
Since $\alpha$ is common to all codewords for a fixed user location, it does not affect their relative comparison. Hence, we use the normalized beamforming gain
\begin{align}
    \nu_k = \lvert \mathbf{a}(\theta_t,r_t)^{\mathsf{H}} \mathbf{w}_{k} \rvert^2, \quad k=1,2,\ldots,K,
\end{align}
as a performance metric for comparing different codebooks. Since both $\mathbf{a}(\theta_t,r_t)$ and $\mathbf{w}_{k}$ are unit-norm, it follows that $0 \leq \nu_k \leq 1$.

For a fixed user location, the random variables $\{\nu_k\}_{k=1}^{K}$ are independent and identically distributed (i.i.d.) realizations of a common random variable $\nu$. Thus, the performance of a random codebook can be characterized through the cumulative distribution function (CDF)
\begin{align}
    F_{\nu}(\gamma) = \Pr\{ \nu \leq \gamma \}, 
    \quad 0 \leq \gamma \leq 1.
\end{align}

The relevant performance metric for the entire codebook is the maximum beamforming gain. Since $\{\nu_{k}\}_{k=1}^{K}$ are i.i.d., by order statistics, 
\begin{align}\label{order_statistics}
    \Pr\left\{ \max_{1 \leq k \leq K} \nu_{k} \leq \gamma \right\} = \left( F_{\nu}(\gamma) \right)^{K}.
\end{align} 
If $F_{\nu}(\gamma) < 1$ for every $\gamma < 1$, then
\begin{align}
    \lim_{K \to \infty} \Pr\left\{ \max_{1 \leq k \leq K} \nu_{k} \leq \gamma \right\} = 0, \quad 0 \leq \gamma < 1,
\end{align}
which implies that $\max_{1 \leq k \leq K} \nu_{k}$ converges in probability to one as $K \to \infty$.

The rate of this convergence is determined by the high-gain tail of the single codeword CDF. Let $\epsilon = 1 - \gamma$, and
\begin{align}
p_{M}(\epsilon) \triangleq \Pr\{ \nu \geq 1 - \epsilon \} = 1 - F_{\nu}(1 - \epsilon).
\end{align}
Then,
\begin{align}
    \Pr\left\{ \max_{1 \leq k \leq K} \nu_{k} \leq 1 - \epsilon \right\} = \left( 1 - p_{M}(\epsilon) \right)^K.
\end{align}
For small $p_{M}(\epsilon)$, we have $\left( 1 - p_{M}(\epsilon) \right)^K \approx e^{-K p_{M}(\epsilon)}$. Thus, for a target gap $\epsilon \ll 1$, the probability that the maximum gain exceeds $1 - \epsilon$ approaches one when $K p_{M}(\epsilon) \to \infty$.
Therefore, the right-tail behavior of $F_{\nu}(\gamma)$ determines how large the codebook must be to obtain near-unit beamforming gain.
\section{Performance Analysis of Random Codebooks}\label{random_codebooks}
\subsection{Random Vector Quantization (RVQ)}\label{RVQ}
We first consider the RVQ codebook \cite{1290092}, where each codeword is generated independently and isotropically over the complex unit sphere in $\mathbb{C}^{M}$. Equivalently, each RVQ codeword can be constructed by first selecting an unnormalized random vector $\mathbf v_k \sim \mathcal{CN}(\mathbf 0, \mathbf I_M)$,
and then normalizing it as $\mathbf w_k = \frac{\mathbf v_k}{\lVert \mathbf v_k \rVert_{2}}$. Since the complex Gaussian vector $\mathbf v_k$ is isotropically distributed, the normalized vector $\mathbf w_k$ is uniformly distributed on the complex unit sphere. The RVQ codebook is then given by $\mathcal{F} = \{ \mathbf{w}_1, \mathbf{w}_2,\ldots,\mathbf{w}_{K} \}$, where the $K$ codewords are generated independently.

RVQ has been widely studied in the beamforming and limited-feedback literature due to its analytical tractability and strong asymptotic properties \cite{1365326}. In particular, Santipach and Honig showed that independently generated RVQ codebooks are asymptotically optimal for beamforming in the large-system regime\cite{4787627}. Therefore, RVQ is commonly used as a benchmark for evaluating structured codebooks.

For the considered model, let $\nu = \lvert \mathbf{a}(\theta_t,r_t)^{\mathsf{H}} \mathbf{w} \rvert^2$, where $\mathbf{w}$ is an RVQ codeword. Since $\mbf{w}$ is isotropically distributed , the distribution of $\nu$ is independent of $\mathbf{a}(\theta_t,r_t)$. The CDF of $\nu$ is given by \cite{4100151} 
\begin{align}\label{cdf_rvq}
    F_{\nu}(\gamma) = 1 - (1-\gamma)^{M-1}, \quad 0 \leq \gamma \leq 1.
\end{align}
 
By standard results from order statistics, the CDF of the maximum beamforming gain, $\max_{1 \leq k \leq K} \nu_{k}$, is given by 
\begin{align}
    \text{Pr}\{ \max_{1 \leq k \leq K} \nu_k \leq \gamma \} = ( 1 - (1-\gamma)^{M-1} )^{K}.
\end{align}
For a fixed $M$ and any $\gamma \in [0,1)$, we have $\lim_{K \to \infty} \text{Pr}\{ \max_{1 \leq k \leq K} \nu_k \leq \gamma \} = 0$.
Thus, for fixed $M$, the maximum beamforming gain converges in probability to one as the codebook size increases, i.e., 
\begin{align}
    \max_{1 \leq k \leq K} \nu_{k} \xrightarrow[K \to \infty]{\mathsf{p}} 1.
\end{align}
\subsection{Random Phase (RP)}
We next consider the RP codebook. Unlike RVQ, whose codewords have unequal magnitudes across the antenna elements \cite{4787627}, the RP codebook imposes a constant-modulus structure on every beamforming vector \cite{7842197}. 
This is well matched to the near-field steering vector $\mathbf{a}(\theta_t,r_t)$, whose entries have equal magnitude $1/\sqrt{M}$, with the user location affecting only the phase profile.

In the RP codebook, the codebook $\mathcal{F}$ is constructed by generating each of the $K$ beamforming vectors independently, where the phase of each antenna is drawn independently and uniformly over $[0,2\pi)$. Specifically, the $k$-th codeword is given by
\begin{align}
    \mathbf{w}_k = \frac{1}{\sqrt{M}}\left[ e^{j\phi_{0,k}},e^{j\phi_{1,k}},\,\ldots,\,e^{j\phi_{M-1,k}} \right]^{\mathsf{T}},
\end{align}
where $\phi_{n,k} \sim \textrm{i.i.d.  Unif}[0,2\pi), \forall n,k$.

For a user with a near-field array response $\mbf{a}(\theta_t,r_t)$, the beamforming gain associated with the \(k\)th RP codeword is
\begin{align}\label{define_gain_RP}
    \nu_k
    = \lvert\mathbf{a}(\theta_t,r_t)^{\mathsf{H}}\mathbf{w}_k\rvert^2
    = \frac{1}{M^2}\left\lvert\sum_{n=0}^{M-1} e^{j\delta_{n,k} }\right\rvert^2,
\end{align}
where $\delta_{n,k} = \varphi_n - \phi_{n,k}$. Since $\phi_{n,k}$ is uniformly distributed over $[0,2\pi)$, the phase difference $\delta_{n,k}$ is also uniformly distributed over $[0,2\pi)$. Hence, for any fixed user location $(\theta_t,r_t)$, the distribution of $\nu_{k}$ is independent of the particular near-field phase profile and is determined only by the coherent sum of $M$ independent unit-magnitude phasors.

The exact distribution of $\nu_k$ is more difficult to characterize than in the RVQ case because $\nu_k$ depends on the squared magnitude of a sum of independent unit phasors. Nevertheless, useful analytical expressions can be obtained in two important regimes. The first is the large-array regime, where the real and imaginary parts of the sum can be approximated as Gaussian random variables using the central limit theorem. This regime captures the bulk behavior of the RP gain distribution. The second is the right-tail regime near $\nu_k = 1$, which corresponds to the rare events where the random phases are nearly aligned. The right-tail regime is especially important when the codebook size is large. We summarize these two regimes in the following lemmas.
\begin{lemma}\label{lemma_-1}
    For an RP codeword $\mbf{w}_{k}$ and large $M$, the scaled gain $M\nu_{k}$ converges in distribution to an exponential random variable with unit rate. Therefore, the CDF of each $\nu_{k}$ satisfies
    \begin{align}
        F_{\nu}(\gamma) = 1 - e^{-M \gamma} + \mathcal{O}\left(\frac{1}{\sqrt{M}}\right).
    \end{align}
\end{lemma}

\begin{proof}
    See Appendix \ref{appendix_-1}.
\end{proof}
\begin{lemma}\label{lemma_rp_high_gain}
    For $\gamma \in [1 - \epsilon,1]$ with $M^2\epsilon \ll 1$, the CDF of the RP beamforming gain satisfies 
    \begin{align}
        F_{\nu}(\gamma) = 1 - C_{M}(1-\gamma)^{\frac{M-1}{2}} + \mathcal{O}((1-\gamma)^{\frac{M+1}{2}}),
    \end{align}
    where $C_{M} = \frac{M^{\frac{M}{2}}}{2^{M-1} \pi^{\frac{M-1}{2}} \Gamma\left(\frac{M+1}{2}\right)}$ is a constant depending only on $M$ and $\Gamma(\cdot)$, which is the Gamma function.
\end{lemma}
\begin{proof}
    See Appendix \ref{appendix_rp_highgain}.
\end{proof}
Lemma \ref{lemma_-1} provides a large-$M$ approximation for the distribution of the RP beamforming gain, obtained using the Berry–Esseen theorem\cite{vershynin2026friendly}. It accurately captures the bulk behavior of the RP gain distribution for thresholds away from the upper endpoint, $\gamma \in [0,1-\epsilon]$. However, this approximation becomes less accurate near the upper end of the support because the true gain satisfies $\nu \in [0,1]$, whereas the exponential approximation does not explicitly enforce this bounded support. In particular, for finite $M$, the exponential approximation does not satisfy $F_{\nu}(1) = 1$.

To address this limitation, Lemma \ref{lemma_rp_high_gain} characterizes the right-tail behavior of the distribution when the gain is close to its maximum possible value, i.e., $\gamma \in [1-\epsilon, 1]$. The resulting polynomial approximation is particularly useful in the high-gain regime because it follows the support $[0,1]$, remains accurate near $\gamma = 1$, and does not rely on $M$ being large. This regime is important for large codebooks since the maximum gain is governed by rare high-gain events.

For the RP codebook, the CDF of the maximum beamforming gain over $K$ independently generated codewords is
\begin{align}
    \Pr\left\{ \max_{1 \leq k \leq K} \nu_k \leq \gamma \right\} = \left(F_{\nu}(\gamma)\right)^K.
\end{align}
Since the RP gain satisfies $0 \leq \nu_{k} \leq 1$, we have $F_{\nu}(1) = 1$. Moreover, because the phases are random, for any fixed $\gamma \in [0,1)$, there is a non-zero probability that an RP codeword achieves a gain larger than $\gamma$. Since $F_{\nu}(\gamma) < 1$ for every fixed $\gamma < 1$, increasing $K$ decreases $F_{\nu}(\gamma)^K$. Equivalently, a large codebook provides more independent phase realizations, increasing the probability that at least one codeword achieves a gain larger than $\gamma$. Therefore,
\begin{align}
    \lim_{K \to \infty}
    \Pr\left\{ \max_{1\leq k\leq K} \nu_k \leq \gamma \right\}
    = 0.
\end{align}
Hence, for fixed $M$, the maximum beamforming gain converges in probability to one as $K \to \infty$, i.e.,
\begin{align}
    \max_{1\leq k\leq K} \nu_k
    \xrightarrow[K\to\infty]{\mathsf{p}} 1.
\end{align}
\subsection{Random Far-Field (RFF) Codebook}\label{RFF}
We next consider the RFF codebook, whose codewords are randomly selected far-field steering vectors. Unlike the RP codebook, where antenna phases are generated independently, each RFF codeword follows the structured linear phase progression of a far-field beam. Thus, RFF provides a useful benchmark for evaluating how conventional far-field beams perform when used to serve near-field users. 
Although a far-field beam can match the linear phase component of the near-field channel, it cannot eliminate the quadratic phase mismatch associated with finite-distance propagation.

In the RFF codebook, each beamforming vector is generated by selecting a random far-field direction and forming the corresponding far-field steering vector. Specifically, the $k$-th codeword is given by
\begin{align}
    \mathbf{w}_k = \mathbf{a}(\theta_k,\infty),
\end{align}
where $\theta_k$ is drawn independently from a uniform distribution over $[-1,1]$. Hence, the codebook consists of $K$ independent far-field beams pointing in random directions over the normalized angular domain.

For a user located at $(\theta_t,r_t)$, the beamforming gain achieved by the $k$-th codeword is
\begin{align}\label{Y_k_original}
    \nu_k
    = \left\lvert
    \frac{1}{M}
    \sum_{n=0}^{M-1}
    e^{j\frac{2\pi}{\lambda}\left(
    -nd\theta_t + \frac{n^2d^2}{2}\frac{1-\theta_t^2}{r_t}
    \right)}
    e^{-j\frac{2\pi}{\lambda}(-nd\theta_k)}
    \right\rvert^2.
\end{align}
This expression makes the source of the mismatch clear. The far-field codeword contributes only a linear phase progression in $n$, whereas the near-field channel contains both linear and quadratic phase terms. To simplify the notation, define
\begin{align}
    \psi_n(\theta_k,\theta_t,r_t)
    \triangleq
    \frac{2\pi}{\lambda}
    \left(
    nd(\theta_k-\theta_t)
    + \frac{n^2d^2}{2}\frac{1-\theta_t^2}{r_t}
    \right).
\end{align}
Then, the gain can be written as\footnote{We use $\nu_k$ and $\nu_k(\theta_k,\theta_t, r_t)$ interchangeably in Section \ref{RFF}.}
\begin{align}\label{lemma_eqn}
    \nu_k(\theta_k,\theta_t,r_t)
    = \frac{1}{M^2}
    \sum_{n=0}^{M-1}\sum_{m=0}^{M-1}
    e^{j( \psi_n(\theta_k,\theta_t,r_t) - \psi_m(\theta_k,\theta_t,r_t) ) }.
\end{align}
Let $ \ell = n - m$, we obtain
\begin{align}
    \psi_n(\theta_k,\theta_t,r_t) - \psi_m&(\theta_k,\theta_t,r_t) = \pi \ell (\theta_k-\theta_t) \notag \\
    &+ \pi (2m\ell+\ell^2) \frac{\lambda}{4}\frac{1-\theta_t^2}{r_t}.
\end{align}
This expression shows that the near-field effect enters the gain through the location-dependent \emph{quadratic phase curvature parameter}\footnote{We use $A$ and $A(\theta_t,r_t)$ interchangeably.} \begin{align}\label{define_A}
    A(\theta_t,r_t) \triangleq \frac{\lambda}{4}\frac{1-\theta_t^2}{r_t}.
\end{align}
The parameter $A(\theta_t,r_t)$ is dimensionless and captures the quadratic phase component of the near-field channel. Larger values of $A(\theta_t,r_t)$ correspond to stronger spherical-wave phase curvature across the array, while $A(\theta_t,r_t) \to 0$ recovers the far-field phase model. $A(\theta_t,r_t)$ is the largest near the array broadside, where $\theta_t \approx 0$, and it decreases as the user moves closer to the array-axis direction $\lvert \theta_t \rvert \approx 1$.

For users distributed over $\theta_t \in [\theta_{\min}, \theta_{\max}]$ and $r_t \in [ r_{\min}, r_{\max} ]$, the range of $A(\theta_t,r_t)$ is $A_{\min} \leq A(\theta_t,r_t) \leq A_{\max}$. Since $\theta_t \in [-1,1]$ and $r_t \in [r_{\min}, r_{\max}]$, we have $A_{\min} = 0$ and $A_{\max} = \frac{\lambda}{4 r_{\min}}$. Using $D = Md$, we can rewrite $r_{\min}$ as $r_{\min} = \frac{0.62 \lambda M}{ 2 } \sqrt{\frac{M}{2}}$. Substituting this into $A_{\max} = \lambda/(4 r_{\min})$ gives $A_{\max} = \frac{1}{1.24 M} \sqrt{\frac{2}{M}}$. Hence, 
\begin{align}\label{define_A_range}
    A(\theta_t,r_t) \in \left[0 , \frac{1}{1.24 M} \sqrt{\frac{2}{M}} \right].
\end{align}

\begin{lemma}\label{lemma_0}
The gain function $\nu_{k}$ can be expressed as
\begin{align}\label{nu_A_B}
        \nu_{k}(\theta_k,\theta_t,r_t) &= \frac{1}{M} + \frac{2}{M^2} \Bigg( \sum_{p = 1}^{M-1} \frac{\sin\big( \pi(M-p)A p \big)}{\sin( \pi A p )} \notag \\
        &\quad \cos \Big(\pi p \big( (M-1)A(\theta_t,r_t)+ \theta_k - \theta_t \big) \Big)  \Bigg).
\end{align}
\end{lemma}
\begin{proof}
    See Appendix \ref{appendix_0}.
\end{proof}
\subsubsection{Maximum achievable gain}
We now study the largest beamforming gain that can be achieved by any far-field beam for a user located at $(\theta_t,r_t)$. The exact maximum gain over all far-field directions is
\begin{align}
\nu_{\max}(\theta_t,r_t)
\triangleq
\max_{\theta_k \in [-1,1]}
\nu_k(\theta_k,\theta_t,r_t).
\end{align}
This quantity characterizes the best possible performance attainable by a far-field beam when serving a near-field user. However, $\nu_k(\theta_k,\theta_t,r_t)$ is a nonconvex trigonometric polynomial in $\theta_k$ and generally contains multiple stationary points. Therefore, obtaining a closed-form expression for all stationary points is analytically intractable. Instead, we characterize a symmetry-based family of approximate stationary points and solve a related problem: maximizing the gain over this candidate family.

Differentiating \eqref{nu_A_B} with respect to $\theta_k$ gives
\begin{align}\label{Y_k_derivative}
\frac{\partial \nu_k(\theta_k,\theta_t,r_t)}{\partial \theta_k}
&=
\frac{-2\pi}{M^2}
\sum_{p=1}^{M-1}
p
\frac{\sin(\pi(M-p)Ap)}{\sin(\pi Ap)}
\notag\\
& \mspace{-18mu}\mspace{-18mu}\mspace{-9mu} \times
\sin\Big(\pi p\big((M-1)A(\theta_t,r_t)+\theta_k-\theta_t\big)\Big).
\end{align}
Define
\begin{align}\label{define_g}
g_p(\theta_k, \theta_t, r_t) &\triangleq \sin\Big(\pi p\big((M-1)A(\theta_t,r_t)+\theta_k-\theta_t\big)\Big), \\
\label{define_f}
f_p(\theta_t, r_t) &\triangleq p\frac{\sin(\pi(M-p)Ap)}{\sin(\pi Ap)}.
\end{align}
Then,
\begin{align}\label{define_derivative}
\frac{\partial \nu_k(\theta_k,\theta_t,r_t)}{\partial \theta_k}
=
\frac{-2\pi}{M^2}
\sum_{p=1}^{M-1}
f_p(\theta_t, r_t) g_p(\theta_k, \theta_t, r_t).
\end{align}
\begin{lemma}\label{lemma_1}
For $M\gg1$, we have
\begin{align}\label{define_f_approximation}
f_p(\theta_t, r_t)
=
\frac{1}{\pi A}
\left(1+\mathcal O\left(\frac{1}{M}\right)\right)
\sin(\pi(M-p)Ap),
\end{align}
for all $p\in\{1,2,\ldots,M-1\}$. Thus, $f_p(\theta_t, r_t)$ is approximately even about $p=M/2$.
\end{lemma}

\begin{proof}
See Appendix \ref{appendix_A}.
\end{proof}

Using Lemma \ref{lemma_1}, the derivative $\partial \nu_k/\partial\theta_k$ approximately vanishes when $g_p(\theta_k,\theta_t,r_t)$ is odd about $p=M/2$. This gives the following family of approximate stationary points.
\begin{lemma}\label{lemma_2}
Let $\theta_k^{(b)}$ denote a symmetry-based approximate stationary far-field direction of $\nu_k(\theta_k,\theta_t,r_t)$, where $b\in\mathbb Z$ is an integer index. For large $M$,
\begin{align}
\theta_k^{(b)}
=
\theta_t-(M-1)A(\theta_t,r_t)+\frac{2b}{M},
\qquad b\in\mathbb Z.
\end{align}
\end{lemma}

\begin{proof}
See Appendix \ref{appendix_B}.
\end{proof}

The candidate directions $\theta_k^{(b)}$ are centered around a location-dependent direction. We define this \emph{location-dependent center direction}  as $\theta_c(\theta_t,r_t) \triangleq \theta_t-(M-1)A(\theta_t,r_t)$. Substituting $\theta_k^{(b)}$ into \eqref{nu_A_B} gives $\nu_k(\theta_k^{(b)},\theta_t,r_t) = \frac{1}{M} + \frac{2}{M^2} \sum_{p=1}^{M-1} \cos\left(\pi p\frac{2b}{M}\right) \frac{\sin(\pi(M-p)Ap)}{\sin(\pi Ap)}$.
Since $\cos(\cdot)$ is even, $\cos\left(\pi p\frac{2(-b)}{M}\right) =
\cos\left(\pi p\frac{2b}{M}\right)$, and therefore $\nu_k(\theta_k^{(-b)},\theta_t,r_t) = \nu_k(\theta_k^{(b)},\theta_t,r_t)$.
Thus, the candidates indexed by $b$ and $-b$ achieve the same gain. Hence, for the gain value, it is sufficient to consider only $b\in\mathbb N_0$.

For $b\in\mathbb N_0$, define the candidate gain
\begin{align}\label{gain_define_critical_points}
G_b(A)
\triangleq
\frac{1}{M}
+
\frac{2}{M^2}
\sum_{p=1}^{M-1}
\cos\left(\pi p\frac{2b}{M}\right)
\frac{\sin(\pi(M-p)Ap)}{\sin(\pi Ap)}.
\end{align}
We now define the candidate-direction approximation to the maximum gain as
$\nu_{\max}^{\mathrm{cand}}(\theta_t,r_t)
\triangleq
\max_{b\in\mathbb N_0}
G_b(A(\theta_t,r_t))$.
This quantity is not obtained by solving for all stationary points of the continuous function $\nu_k(\theta_k,\theta_t,r_t)$. Rather, it is obtained by maximizing the exact gain over the analytically characterized candidate directions $\theta_k^{(b)}$.

To determine which candidate direction gives the largest gain, we introduce a sequence $\{A_{\textrm{th},b}\}_{b\in\mathbb N_0}$ with $A_{\textrm{th},0}=0$, where the subscript ``th'' denotes a threshold point. To solve for $A_{\textrm{th},b}$, we treat $A$ as a scalar variable. Then, $A_{\textrm{th},b}$ is the smallest value of $A$ at which the gains achieved by the candidate directions $\theta_k^{(b-1)}$ and $\theta_k^{(b)}$ are equal. For each positive integer $b$, $A_{\textrm{th},b}$ is defined as the smallest positive solution of the following equation
\begin{align}\label{define_B}
&\sum_{p=1}^{M-1}
\cos\left(\pi p\frac{2(b-1)}{M}\right)
\frac{\sin(\pi(M-p)pA_{\textrm{th},b})}
{\sin(\pi pA_{\textrm{th},b})}
\notag \\
&\quad =
\sum_{p=1}^{M-1}
\cos\left(\pi p\frac{2b}{M}\right)
\frac{\sin(\pi(M-p)pA_{\textrm{th},b})}
{\sin(\pi pA_{\textrm{th},b})}.
\end{align}

For a given user location, $A(\theta_t,r_t)$ is first computed, and $b^\star$ is selected as the index of the threshold interval containing $A(\theta_t,r_t)$, i.e.,
$A_{\textrm{th},b^\star}
\leq
A(\theta_t,r_t)
\leq
A_{\textrm{th},b^\star+1}$.

\begin{lemma}\label{lemma_3}
For large $M$, the maximizing candidate directions are
\begin{align}
\theta_k^{(b^\star)}
=
\theta_t-(M-1)A(\theta_t,r_t)
\pm
\frac{2b^\star}{M},
\end{align}
and the corresponding candidate-direction approximation to the maximum gain is\footnotemark
\begin{align}\label{max_ff_cb}
\nu_{\max}^{\mathrm{cand}}(\theta_t,r_t) = G_{b^\star}(A(\theta_t,r_t)).
\end{align}
\end{lemma}
\footnotetext{$\nu_{\max}^{\mathrm{cand}}(\theta_t,r_t)$ is as a function of $A(\theta_t,r_t)$.}
\begin{proof}
See Appendix \ref{appendix_C}.
\end{proof}
\begin{remark}\label{remark_execption}
    The interval rule identifying $b^\star$ in Lemma \ref{lemma_3} is exact except on a narrow range of $A(\theta_t,r_t)$ near the transition between $b^\star=2$ and $b^\star =3$, where the true maximizing candidate is instead $b^\star = 0$. The resulting loss in beamforming gain from using the stated rule is at most $\approx2\%$. This exceptional range has width $\mathcal{O}(1/M^2)$ in $A(\theta_t,r_t)$ and therefore shrinks as the array size grows, consistent with the large-$M$ regime the lemma targets. Since the loss is small, we retain the simpler, general form of Lemma \ref{lemma_3} and use it as stated throughout the paper. The exact range is given in Appendix \ref{appendix_C}.
\end{remark}

Equation \eqref{max_ff_cb} gives an analytical approximation of $\nu_{\max}(\theta_t,r_t)$. The exact maximum $\nu_{\max}(\theta_t,r_t)$ is obtained by maximizing $\nu_k(\theta_k,\theta_t,r_t)$ over the continuous far-field direction $\theta_k$, whereas $\nu_{\max}^{\mathrm{cand}}(\theta_t,r_t)$ is obtained by maximizing over the candidate directions $\theta_k^{(b)}$. Fig. \ref{fig:1} compares $\nu_{\max}^{\mathrm{cand}}(\theta_t,r_t)$ with direct numerical maximization over $\theta_k$. The close agreement for all considered values of $M$ confirms the accuracy of the candidate-direction approximation.


For near-field users, the maximum gain is generally below one because no far-field beam can eliminate the residual quadratic phase mismatch. Hence, as the RFF codebook size increases, the gain approaches $\nu_{\max}(\theta_t,r_t)$, rather than one. For small values of $A(\theta_t,r_t)$, the maximum gain is close to one. This is expected because a small value of $A(\theta_t,r_t)$ implies that the quadratic component of the near-field is small. As a result, the residual phase mismatch between the near-field channel and the best far-field beam is small, so a far-field beam can closely match the channel. As $A(\theta_t,r_t)$ increases, the quadratic phase component becomes stronger and the maximum achievable gain decreases.

Fig. \ref{fig:1} also shows that the degradation occurs at smaller values of $A(\theta_t,r_t)$ as $M$ increases. This is because the quadratic phase scales approximately as $A(\theta_t,r_t)M^2$. Hence, for larger arrays, the mismatch between the near-field channel and the far-field beams becomes significant even for smaller values of $A(\theta_t,r_t)$. The figure also shows several transition points, which correspond to the discrete sequence $\{ A_{\textrm{th},i} \}$. The transition points occur at similar gain levels but at different values of $A(\theta_t,r_t)$ for different array sizes. A more detailed characterization of these transition points and their relationship to the far-field gain profile is left as future work.

Figs. \ref{fig:2}, \ref{fig:3}, and \ref{fig:4} show heat maps of $\nu_{\max}(\theta_t,r_t)$ when the codebook is restricted to far-field beams while the user channel is near-field. The array is placed along the X-axis, and the gain is plotted over the user location in the X-Y plane. As $M$ increases, the maximum achievable gain deteriorates over a larger portion of the user region. In particular, Fig. \ref{fig:4} shows that for $M = 256$, the loss remains noticeable even for user distances on the order of 100 m. This is a consequence of the stronger near-field effect induced by the larger array aperture, which makes the mismatch between the near-field channel and the far-field beam more significant. 

The degradation is most severe near the array broadside, where the quadratic phase variation across the array is strongest. In contrast, when the user is located close to the array axis direction, i.e., near $\lvert\theta_t\rvert = 1$, the maximum gain remains high. This behavior is consistent with the quadratic phase curvature parameter $A(\theta_t,r_t)$, which is largest near broadside and vanishes when $\theta_t = \pm 1$. Therefore, for $\theta_t \approx \pm 1$, the near-field quadratic phase term is minimal, and the channel behaves similarly to a far-field channel.
\begin{figure}[h]
    \centering
    \includegraphics[width=0.8\columnwidth]{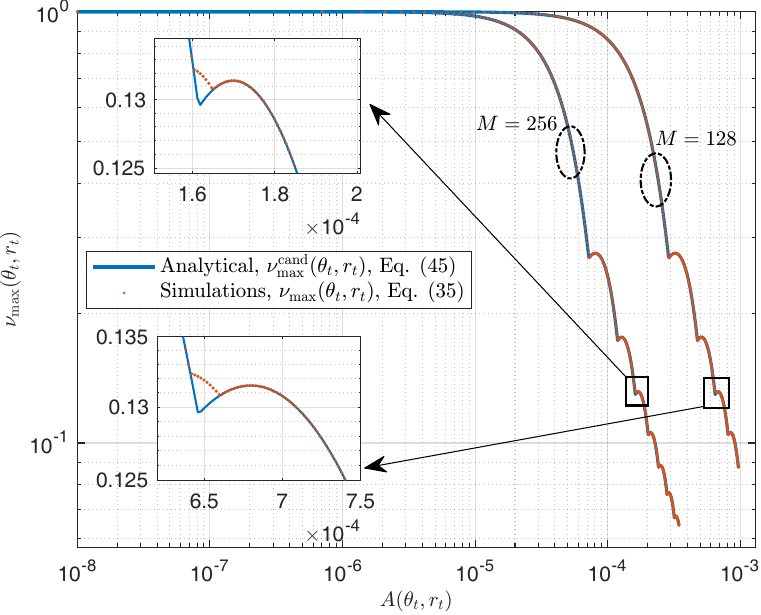}
    \caption{$\nu_{\text{max}}(\theta_t,r_t)$ as a function of $A(\theta_t,r_t)$.}
    \label{fig:1}
\end{figure}
\begin{figure}[h]
    \centering
    \includegraphics[width=0.8\columnwidth]{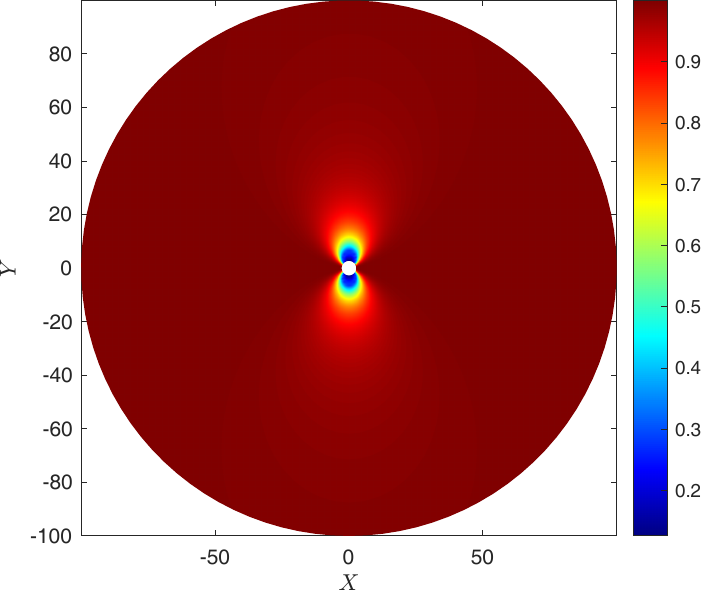}
    \caption{Heat map of $\nu_{\text{max}}(\theta_t,r_t)$ for $M = 64$.}
    \label{fig:2}
\end{figure}
\begin{figure}[h]
    \centering
    \includegraphics[width=0.8\columnwidth]{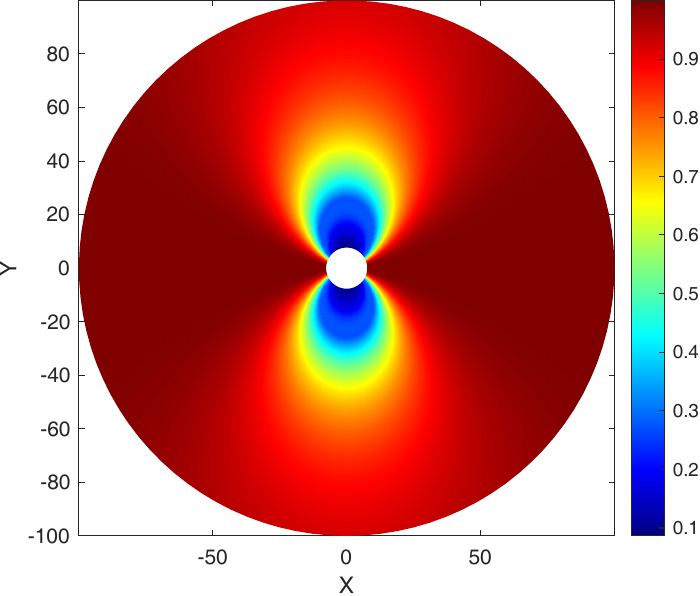}
    \caption{Heat map of $\nu_{\text{max}}(\theta_t,r_t)$ for $M = 128$.}
    \label{fig:3}
\end{figure}
\begin{figure}[h]
    \centering
    \includegraphics[width=0.8\columnwidth]{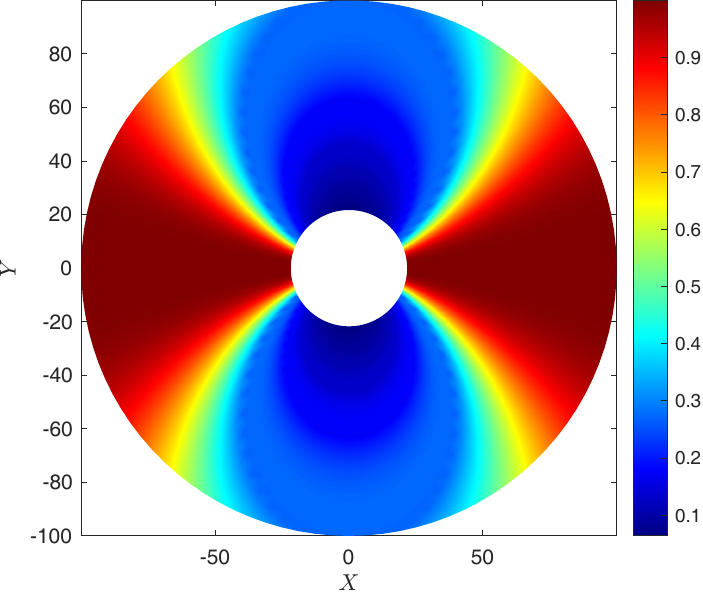}
    \caption{Heat map of $\nu_{\text{max}}(\theta_t,r_t)$ for $M = 256$.}
    \label{fig:4}
\end{figure}
\subsubsection{Asymptotic CDF of \texorpdfstring{$\max_{k} \nu_{k}$}{max k nu k}}
We now characterize the distribution of the best gain among $K$ independently generated RFF codewords. For a user located at $(\theta_t,r_t)$, let
\begin{align}
    F_{\nu\mid \theta_t,r_t}(\gamma)
    \triangleq
    \Pr\{\nu \le \gamma \mid \theta_t,r_t\}
\end{align}
denote the conditional CDF of the gain achieved by a single randomly generated far-field beam. Since the $K$ codewords are generated independently, the corresponding gains $\{\nu_k\}_{k =1}^{K}$ are conditionally i.i.d., and therefore $\Pr \{ \max_{ k} \nu_k \le \gamma \mid \theta_t,r_t \} = (F_{\nu\mid\theta_t,r_t}(\gamma))^K$. Averaging over the user-location distribution yields
\begin{align}\label{cdf_ff_1}
    \Pr\{\max_{ 1\leq k \leq K} &\nu_k \le \gamma \} = \notag \\
    &\iint
    \big(F_{\nu\mid\theta_t,r_t}(\gamma)\big)^K
    f_\Theta(\theta_t)f_{R_t}(r_t)\, d\theta_t\,dr_t.
\end{align}

To evaluate \eqref{cdf_ff_1}, it is useful to distinguish between user locations for which the threshold $\gamma$ is unachievable by any far-field beam and those for which it is achievable. Define 
\begin{align}
    \mathcal R(\gamma)
    \triangleq
    \{(\theta_t,r_t): \nu_{\max}(\theta_t,r_t)\le \gamma\}.
\end{align}
The set $\mathcal{R}(\gamma)$ contains the user locations for which even the best far-field beam achieves a gain of at most $\gamma$. For such locations, $F_{\nu\mid\theta_t,r_t}(\gamma)=1$. Splitting the integral in \eqref{cdf_ff_1} over $\mathcal{R}(\gamma)$ and its complement $\mathcal{R}^{c}(\gamma)$ gives
\begin{align}\label{cdf_ff_2}
    \Pr&\{\max_{1 \leq k \leq K} \nu_k \le \gamma\}
    =
    \iint_{\mathcal R(\gamma)}
    f_\Theta(\theta_t)f_{R_t}(r_t)\,d\theta_t\,dr_t \notag\\
    &\quad+
    \iint_{\mathcal R^c(\gamma)}
    \big(F_{\nu\mid\theta_t,r_t}(\gamma)\big)^K
    f_\Theta(\theta_t)f_{R_t}(r_t)\,d\theta_t\,dr_t.
\end{align}
Denote the first and second integrals in \eqref{cdf_ff_2} by $I_1(\gamma)$ and $I_2(K,\gamma)$, respectively.

For every \((\theta_t,r_t)\in \mathcal R^c(\gamma)\), we have $\nu_{\max}(\theta_t,r_t) > \gamma$. Hence, there exists a non-zero number of far-field directions for which the gain exceeds $\gamma$, and therefore $F_{\nu\mid\theta_t,r_t}(\gamma) < 1$. It follows that $\lim_{K \to \infty} I_2(K,\gamma) = 0$. Thus, from \eqref{cdf_ff_2},
\begin{align}
    \lim_{K\to\infty}\Pr\{\max_{1 \leq k \leq K} \nu_k \le \gamma\}
    =
    I_1(\gamma),
\end{align}
where $I_1(\gamma) = \Pr\left\{ \nu_{\max}(\theta_t,r_t)\le \gamma \right\}$. Therefore, unlike the RVQ and RP codebooks, the RFF codebook does not satisfy
\begin{align}
    \lim_{K\to\infty}\Pr\{\max_{1 \leq k \leq K} \nu_k \le \gamma\}=0
\end{align}
for every $\gamma \in [0,1)$. This occurs because there are user locations for which the threshold $\gamma$ cannot be exceeded by any far-field beam, regardless of the codebook size. This reflects the fundamental residual mismatch between far-field codebooks and near-field channels.
\subsection{Random Near-Field (RNF) Codebook}\label{RNF_section}
We next consider the RNF codebook. Unlike RFF, whose codewords are restricted to far-field steering vectors, RNF codewords are generated directly from the near-field steering vector set. Thus, each RNF codeword contains both the linear and quadratic phase components of the near-field channel, making RNF a natural benchmark for near-field beamforming.

Each RNF codeword $\mathbf{w}_{k} \in \mathcal{F}$ is generated by first sampling a random direction $\theta_k \sim f_{\Theta}(\theta)$ and a random distance $r_k \sim f_{R}(r)$, and then forming the corresponding near-field steering vector. The $k$-th codeword is then given by
\begin{align}
    \mathbf{w}_k = \mathbf{a}(\theta_k,r_k).
\end{align}

For a user with channel vector $\mathbf{h}(\theta_t,r_t) = \alpha \mathbf{a}(\theta_t,r_t)$, the beamforming gain of the $k$-th RNF codeword is
\begin{align}
    \nu_k
    = \left\lvert \mathbf{a}^{\mathsf H}(\theta_t,r_t)\mathbf{a}(\theta_k,r_k) \right\rvert^2.
\end{align}
This gain is large when the sampled pair $(\theta_k,r_k)$ lies close to the true user location $(\theta_t,r_t)$. Therefore, the high-gain behavior of the RNF codebook is governed by the local geometry of the near-field steering vector in the joint angular-distance domain.

The following lemma characterizes the high-gain behavior of a single RNF codeword.
\begin{lemma}\label{lemma_nf}
    For $\gamma = 1 - \epsilon$ with $\epsilon \ll 1$, the CDF of the RNF beamforming gain satisfies
    \begin{align}
        F_{\nu}(\gamma) = 1 - G_{M}(1-\gamma) + \mathcal{O}((1 - \gamma)^{3/2}),
    \end{align}
    where $G_{M}$ is a constant determined by the local near-field steering vector geometry and the sampling distribution.
\end{lemma}
\begin{proof}
    See Appendix \ref{appendix_nf}.
\end{proof}
Lemma \ref{lemma_nf} shows that as $\gamma$ approaches one, the right-tail probability of the RNF gain scales linearly with $1-\gamma$. This linear scaling follows from the fact that RNF codewords are sampled from the same two-parameter near-field model as the channel. Therefore, a high-gain event occurs when the randomly chosen pair $(\theta_k,r_k)$ falls within a sufficiently small neighborhood of $(\theta_t,r_t)$.

Using the order-statistics result from \eqref{order_statistics}, the maximum RNF gain over $K$ independently generated codewords satisfies $\Pr\left\{\max_{1 \leq k \leq K} \nu_k \le \gamma\right\} = \big(F_{\nu}(\gamma)\big)^K$. Since $F_{\nu}(\gamma) < 1$ for every fixed $\gamma < 1$, it follows that 
\begin{align}
    \max_{1\leq k\leq K}\nu_k
    \xrightarrow[K\to\infty]{\mathsf p}
    1.
\end{align}
\section{Simulations}\label{results}
\subsection{Limiting Behavior of the Maximum Beamforming Gain}
We now summarize the limiting behavior of the maximum beamforming gain for the four random codebooks. For RVQ, RP, and RNF, the single-codeword CDF satisfies $F_{\nu}(\gamma) < 1$ for every fixed $\gamma < 1$. Therefore, using the order-statistics result in \eqref{order_statistics}, the maximum gain converges in probability to one as the $K \to \infty$.

Using the general maximum gain characterization in Section \ref{section_codebook_design}, we now compare the high-gain tails of the different codebooks. For RVQ, RP, and RNF, the single codeword high-gain probabilities satisfy $p_{M}^{\mathrm{RVQ}}(\epsilon) = \epsilon^{M-1}$, $p_{M}^{\mathrm{RP}}(\epsilon) \sim C_M \epsilon^{\frac{M-1}{2}}$, and $p_{M}^{\mathrm{RNF}}(\epsilon) \sim G_M \epsilon$, respectively. Since the maximum gain exceeds $1- \epsilon$ with high probability when $K p_M(\epsilon) \to \infty$, the codebook sizes required to achieve $\Pr \{ \max_{k}\ \nu_k \geq 1 - \epsilon \} \to 1$ scale as
\begin{align}
    K_{\mathrm{RVQ}} &\gg \epsilon^{-(M-1)}, \\
    K_{\mathrm{RP}} &\gg \frac{1}{C_M}\epsilon^{-\frac{M-1}{2}}, \\
    K_{\mathrm{RNF}} &\gg \frac{1}{G_M} \epsilon^{-1}.
\end{align}
Thus, RVQ requires exponential growth in $M$ to maintain a fixed target gap $\epsilon$ from unit gain. RP also requires exponential growth in $M$, but with a smaller exponent because the constant-amplitude removes amplitude fluctuations. RNF has the most favorable high-gain scaling since its right-tail probability is linear in $\epsilon$.

In contrast, the RFF codebook is restricted to far-field steering vectors. Increasing the RFF codebook size allows the gain to approach the maximum achievable beamforming gain $\nu_{\max}(\theta_t,r_t)$, but this value is below one for finite-distance near-field users. Therefore, RVQ, RP, and RNF can asymptotically achieve unit gain as $K \to \infty$, whereas RFF saturates below unit gain. 
\subsection{Performance of Different Codebooks}
\begin{figure}[!t]
    \centering
    \includegraphics[width=0.8\columnwidth]{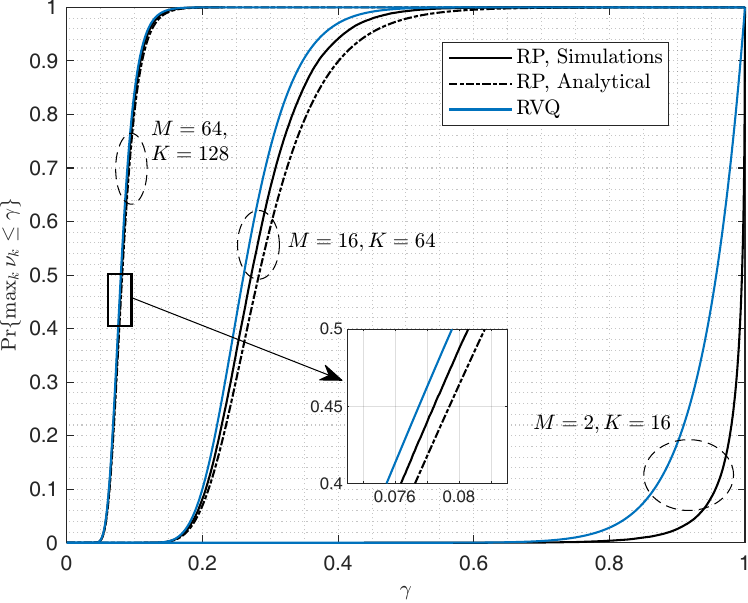}
    \caption{Performance of RVQ and RP codebooks.}
    \label{fig2}
\end{figure}
\begin{figure}[!t]
    \centering
    \includegraphics[width=0.8\columnwidth]{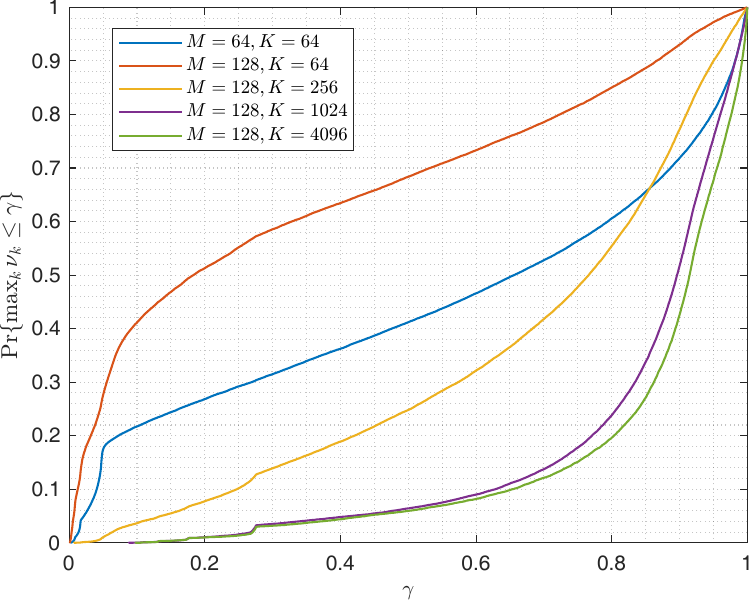}
    \caption{Performance of RFF codebook.}
    \label{fig_rff}
\end{figure}
\begin{figure}[!t]
    \centering
    \includegraphics[width=0.8\columnwidth]{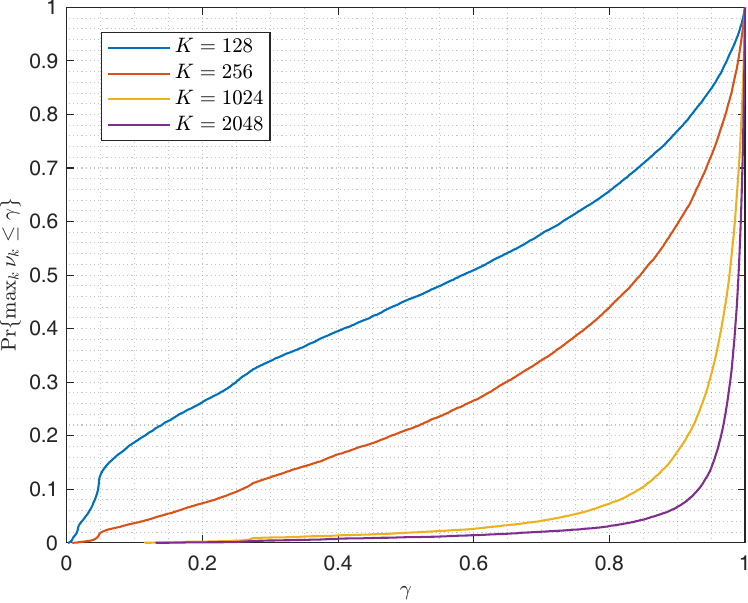}
    \caption{Performance of RNF codebook, $M =128$.}
    \label{fig3}
\end{figure}
\begin{figure}[!t]
    \centering
    \includegraphics[width=0.8\columnwidth]{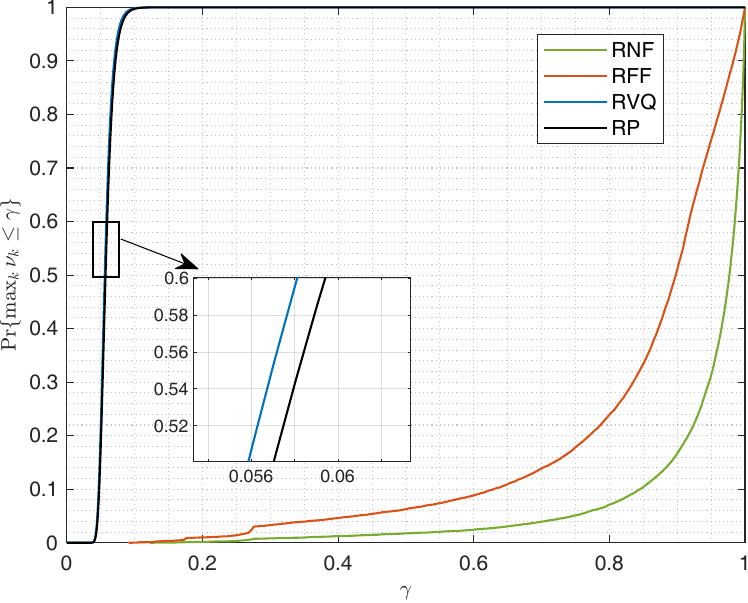}
    \caption{Performance of RVQ, RP, RFF, and RNF codebooks with $M =128$ and $K = 1024$.}
    \label{fig4}
\end{figure}
In this section, we present the simulation results. The operating frequency is 10 GHz. The user lies in a region with $\theta_{t} \in [-1,1]$ and $r_t \in [r_{\min} ,100\ \text{m}]$.

Fig. \ref{fig2} compares the performance of the RVQ and RP codebooks. As the number of antennas increases, the performance of both codebooks degrades significantly. This is because the probability that the phase contributions from all antennas add constructively decreases with $M$, resulting in a lower beamforming gain. Moreover, the RP codebook consistently outperforms the RVQ codebook because the RP codebook preserves the constant-modulus structure across the antenna array and randomizes only the phases. In contrast, RVQ introduces randomness in both amplitudes and phases, which can reduce the effective contribution of some antennas even when the phases are perfectly aligned. This behavior is consistent with the high-gain tail analysis, the RP tail decays as $(1-\gamma)^{\frac{M-1}{2}}$, whereas the RVQ tail decays faster as $(1-\gamma)^{M-1}$. 

Fig. \ref{fig_rff} shows the performance of the RFF codebook. As the codebook size $K$ increases, the maximum beamforming gain improves, which is expected because a larger codebook provides more far-field directions and therefore increases the probability of selecting a beam that is well aligned with the user. However, as the number of antennas $M$ increases, the performance degrades. This is because larger arrays produce narrower far-field beams, so even a small mismatch between the user direction and the selected beam direction can result in lower beamforming gain. In addition, the near-field quadratic phase variation is more significant for larger arrays.

Furthermore, as shown in Section \ref{RFF}, the RFF codebook does not generally achieve unit gain even as $K$ becomes large. Increasing $K$ only allows the codebook to approach the best far-field beam for a given near-field user. Figure \ref{fig_rff} supports this result, since the maximum gain does not equal one even when the number of codewords is large.

Fig. \ref{fig3} shows the performance of the RNF codebook for $M = 128$. As the codebook size $K$ increases, the CDF shifts toward the right, which means that the maximum beamforming gain becomes larger with a higher probability. This is expected because a larger RNF codebook contains more randomly sampled near-field steering vectors, increasing the probability that one of the sampled pairs $(\theta_k,r_k)$ lies close to the true user location $(\theta_t,r_t)$. Moreover, unlike the RFF codebook, as $K$ increases, the maximum beamforming gain converges in probability to one, consistent with the limiting behavior derived in Section \ref{RNF_section}. This is evident from the figure, as $K$ increases, the CDF remains close to zero over most of the gain range and rises sharply near $\gamma = 1$, indicating that the maximum gain is concentrated near one.

Fig. \ref{fig4} compares the CDF of the maximum beamforming gain for the RVQ, RP, RFF, and RNF codebooks. It can be observed that RNF achieves the best performance, followed by RFF, RP, and RVQ. Here, better performance means that the CDF is shifted farther to the right, or equivalently, that the probability $\Pr\{ \max_{k} \nu_k \leq \gamma\}$ is smaller for a given threshold $\gamma$. The RNF codebook provides the best performance, which is intuitive because its codewords capture the near-field phase structure more accurately than those of other codebooks. The RFF codebook performs well for a fixed codebook size because its codewords have a structured phase profile, which is better than RP and RVQ. However, although RFF outperforms RP and RVQ for a fixed $K$, it is asymptotically inferior to both, whereas RNF provides the best performance across all values of $K$.
\section{Conclusions}\label{conclusions}
In this work, we studied four classes of randomly generated codebooks for near-field communications. 
We characterized their beamforming gain distributions and compared the limiting behavior of the maximum gain as the codebook size increases. 
For the RP codebook, we derived a large-array bulk approximation and a right-tail approximation near unit gain. 
For the RFF codebook, we derived an analytical expression for the maximum gain achievable by far-field beams and characterized the corresponding maximizing far-field direction. For the RNF codebook, we characterized the right-tail behavior near unit gain and showed that the probability scales linearly with the gap from unit gain. 
The results show that the performance of RVQ and RP codebooks deteriorates significantly as the array size grows, while RNF provides the most favorable high-gain scaling. RFF can outperfrom RVQ and RP for finite codebook sizes, but its maximum gain saturates below one for finite-distance near-field users. 

Future work includes studying whether far-field beams can be used for low-overhead near-field beam training and whether their measurements contain sufficient information to estimate both the user's angle and distance. Another direction is to characterize the geometry of the near-field high-gain region, which can guide the design of more efficient near-field codebooks.
\bibliography{ref} 
\bibliographystyle{IEEEtran}
\appendices
\section{Proof of Lemma \ref{lemma_-1}}\label{appendix_-1}
Define the random variable $ X = \sum_{n=0}^{M-1} e^{j\delta_n} = X_R+jX_I$, where the real and imaginary components are the random variables $X_R=\sum_{n=0}^{M-1}\cos(\delta_n)$ and $X_I=\sum_{n=0}^{M-1}\sin(\delta_n)$, respectively.

Then, $\nu = \frac{|X|^2}{M^2} = \frac{X_R^2+X_I^2}{M^2}$. Since $\delta_n\sim\mathrm{i.i.d.}\ \mathrm{Unif}[0,2\pi)$, we have $\mathbb E[\cos(\delta_n)]
    =
    \mathbb E[\sin(\delta_n)]
    =
    0$, $\mathrm{Var}(\cos(\delta_n))
    =
    \mathrm{Var}(\sin(\delta_n))
    =
    \frac{1}{2}$, and $\mathbb E[\cos(\delta_n)\sin(\delta_n)]
    =
    0$. Therefore, by the Berry-Esseen theorem \cite{vershynin2026friendly},
\begin{align}
    \left(
    \frac{X_R}{\sqrt{M/2}},
    \frac{X_I}{\sqrt{M/2}}
    \right)
    \xrightarrow{d}
    (Z_1,Z_2),
\end{align}
with an approximation error of order $\mathcal O(1/\sqrt{M})$, where $Z_1$ and $Z_2$ are independent standard normal random variables.

It follows that
\begin{align}
    M\nu
    =
    \frac{1}{2}
    \left[
    \left(\frac{X_R}{\sqrt{M/2}}\right)^2
    +
    \left(\frac{X_I}{\sqrt{M/2}}\right)^2
    \right]
    \xrightarrow{d}
    \frac{1}{2}(Z_1^2+Z_2^2).
\end{align}
Since $Z_1^2+Z_2^2\sim\chi_2^2$, we have $\frac{1}{2}(Z_1^2+Z_2^2) \sim \mathrm{Exp}(1)$. Thus, for large $M$, $M\nu$ is approximately exponentially distributed with unit rate, or equivalently, $\nu$ is approximately exponentially distributed with rate $M$. Hence, $F_\nu(\gamma) = 1-e^{-M\gamma} + \mathcal O\left(\frac{1}{\sqrt{M}}\right)$,
which proves the lemma.
\section{Proof of Lemma \ref{lemma_rp_high_gain}}\label{appendix_rp_highgain}
Let $\epsilon=1-\gamma$. From \eqref{define_gain_RP}, $\nu = \frac{1}{M^2} \left\lvert \sum_{n=0}^{M-1} e^{j\delta_n} \right\rvert^2$, where $\boldsymbol{\delta}=[\delta_0,\ldots,\delta_{M-1}]^{\mathsf T}$
is uniformly distributed over $[0,2\pi)^M$. Expanding the squared
magnitude gives $\nu = 1 - \frac{1}{M^2} \sum_{m,n} 2\sin^2\left(\frac{\delta_m-\delta_n}{2}\right).$
Define the high-gain event $\mathcal D_\epsilon \triangleq \{ \nu \geq 1-\epsilon \}$. Equivalently,
\begin{align}
    \mathcal D_\epsilon
    =
    \left\{
    \sum_{m,n}
    2\sin^2\left(\frac{\delta_m-\delta_n}{2}\right)
    \leq
    M^2\epsilon
    \right\}.
\end{align}
Since all terms in the summation are nonnegative, on
$\mathcal D_\epsilon$ we have $\left| \sin\left(\frac{\delta_m-\delta_n}{2}\right) \right| \leq \frac{M}{\sqrt{2}}\sqrt{\epsilon}, \forall m,n$. Thus, for $M^2\epsilon\ll1$, all phases are close to a common phase
modulo $2\pi$. Equivalently, after unwrapping the phases around this
common phase, $|\delta_m-\delta_n| = \mathcal O(\sqrt{\epsilon}), \forall m,n.$ For any fixed $\rho\in(0,1/2)$ and sufficiently small phase differences,
\begin{align}
    \left(\frac{1}{2}-\rho\right)x^2
    \leq
    2\sin^2\left(\frac{x}{2}\right)
    \leq
    \left(\frac{1}{2}+\rho\right)x^2 .
\end{align}
Therefore, the event $\mathcal D_\epsilon$ is sandwiched as
\begin{align}
    &
    \left\{
    \left(\frac{1}{2}+\rho\right)
    \sum_{m,n}(\delta_m-\delta_n)^2
    \leq
    M^2\epsilon
    \right\}
    \subseteq
    \mathcal D_\epsilon
    \nonumber\\
    &\hspace{1.4cm}
    \subseteq
    \left\{
    \left(\frac{1}{2}-\rho\right)
    \sum_{m,n}(\delta_m-\delta_n)^2
    \leq
    M^2\epsilon
    \right\}.
\end{align}

Now define $\bar{\delta}
    =
    \frac{1}{M}
    \sum_{n=0}^{M-1}\delta_n$ and $\widetilde{\boldsymbol{\delta}}
    =
    \boldsymbol{\delta}
    -
    \bar{\delta}\mathbf 1.$
Using the identity
\begin{align}
    \sum_{m,n}
    (\delta_m-\delta_n)^2
    =
    2M
    \sum_{n=0}^{M-1}
    (\delta_n-\bar{\delta})^2
    =
    2M
    \left\|
    \widetilde{\boldsymbol{\delta}}
    \right\|^2,
\end{align}
the sandwich becomes
\begin{align}
    \left\{
    \left\|
    \widetilde{\boldsymbol{\delta}}
    \right\|^2
    \leq
    \frac{M\epsilon}{1+2\rho}
    \right\}
    \subseteq
    \mathcal D_\epsilon
    \subseteq
    \left\{
    \left\|
    \widetilde{\boldsymbol{\delta}}
    \right\|^2
    \leq
    \frac{M\epsilon}{1-2\rho}
    \right\}.
\end{align}

The centered vector $\widetilde{\boldsymbol{\delta}}$ lies in the
$(M-1)$-dimensional hyperplane $\mathcal H
    =
    \left\{
    \mathbf x\in\mathbb R^M:
    \mathbf 1^{\mathsf T}\mathbf x=0
    \right\}.$ The remaining degree of freedom is the common phase. Since the gain is
invariant to a common phase rotation, this degree of freedom contributes
a factor equal to the length of the diagonal circle, namely
$2\pi\sqrt{M}$. Hence, for a ball of radius $R$ in $\mathcal H$,
\begin{align}
    \Pr\left\{
    \left\|
    \widetilde{\boldsymbol{\delta}}
    \right\|
    \leq R
    \right\}
    =
    \frac{2\pi\sqrt{M}}{(2\pi)^M}
    V_{M-1}(R),
\end{align}
where $V_{M-1}(R)$ is the volume of an $(M-1)$-dimensional Euclidean ball
of radius $R$. Since $V_{M-1}(R) = \frac{\pi^{\frac{M-1}{2}}}{\Gamma\left(\frac{M+1}{2}\right)} R^{M-1}$,we obtain
\begin{align}
    &
    C_M
    (1+2\rho)^{-\frac{M-1}{2}}
    \epsilon^{\frac{M-1}{2}}
    \leq
    \Pr\{\mathcal D_\epsilon\}
    \nonumber\\
    &\hspace{2cm}
    \leq
    C_M
    (1-2\rho)^{-\frac{M-1}{2}}
    \epsilon^{\frac{M-1}{2}},
\end{align}
where $C_M = \frac{M^{\frac{M}{2}}}{2^{M-1}\pi^{\frac{M-1}{2}}\Gamma\left(\frac{M+1}{2}\right)}$, and $\Gamma(\cdot)$ is the Gamma function.
Since $\rho>0$ is arbitrary, letting $\rho\to0$ gives $\Pr\{\mathcal D_\epsilon\}
    \sim
    C_M
    \epsilon^{\frac{M-1}{2}}.$ Using the small-angle expansion $2\sin^2\left(\frac{x}{2}\right)
    =
    \frac{x^2}{2}
    +
    \mathcal O(x^4)$, and the fact that $x=\mathcal O(\sqrt{\epsilon})$ in the high-gain
region, the radius of the high-gain region can be written as $r_\epsilon
    =
    \sqrt{M\epsilon}
    \left(1+\mathcal O(\epsilon)\right).$ Therefore,
\begin{align}
    \Pr\{\nu\geq 1-\epsilon\}
    =
    C_M
    \epsilon^{\frac{M-1}{2}}
    +
    \mathcal O\left(
    \epsilon^{\frac{M+1}{2}}
    \right).
\end{align}

Substituting $\epsilon=1-\gamma$ yields
\begin{align}
    F_\nu(\gamma)
    =
    1
    -
    C_M
    (1-\gamma)^{\frac{M-1}{2}}
    +
    \mathcal O\left(
    (1-\gamma)^{\frac{M+1}{2}}
    \right),
\end{align}
which proves the lemma.
\section{Proof of Lemma \ref{lemma_0}}\label{appendix_0}
After changing variables in \eqref{lemma_eqn} using \(\ell=n-m\), the beamforming gain can be written as
\begin{align}
    \nu_k =
    \frac{1}{M^2}
    \sum_{\ell=-M+1}^{M-1}
    \sum_{m=\max\{0,-\ell\}}^{\min\{M-1,M-1-\ell\}}
    e^{j\pi\left(A(2m\ell+\ell^2)+(\theta_k-\theta_t)\ell\right)}.
\end{align}
We separate the summation into the three cases \(\ell<0\), \(\ell=0\), and \(\ell>0\). For \(\ell=0\), the exponential term is equal to \(1\), and hence the inner sum contributes $\sum_{m=0}^{M-1} 1 = M$.
Therefore, $\nu_k = \frac{\zeta_1 + M + \zeta_2}{M^2}$, where
\begin{align}
    \zeta_1
    &\triangleq
    \sum_{\ell=-M+1}^{-1}
    \sum_{m=-\ell}^{M-1}
    e^{j\pi\left(A(2m\ell+\ell^2)+\ell(\theta_k-\theta_t)\right)},\\
    \zeta_2
    &\triangleq
    \sum_{\ell=1}^{M-1}
    \sum_{m=0}^{M-1-\ell}
    e^{j\pi\left(A(2m\ell+\ell^2)+\ell(\theta_k-\theta_t)\right)}.
\end{align}
We first simplify \(\zeta_1\). Letting \(\ell \mapsto -\ell\), we obtain
\begin{align}
    \zeta_1
    =
    \sum_{\ell=1}^{M-1}
    e^{j\pi\left(A\ell^2-\ell(\theta_k-\theta_t)\right)}
    \sum_{m=\ell}^{M-1} e^{-j2\pi A\ell m}.
\end{align}
The inner sum is geometric, and applying the finite geometric-series identity yields
\begin{align}
    \sum_{m=\ell}^{M-1} e^{-j2\pi A\ell m}
    =
    e^{-j\pi A\ell(M+\ell-1)}
    \frac{\sin\big(\pi(M-\ell)A\ell\big)}{\sin(\pi A\ell)}.
\end{align}
Substituting this into the expression for \(\zeta_1\) gives
\begin{align}\label{delta_1_simplify_rewrite}
    \zeta_1
    =
    \sum_{\ell=1}^{M-1}
    e^{-j\pi \ell\big(A(M-1)+\theta_k-\theta_t\big)}
    \frac{\sin\big(\pi(M-\ell)A\ell\big)}{\sin(\pi A\ell)}.
\end{align}
Similarly, for \(\zeta_2\),
\begin{align}
    \zeta_2
    =
    \sum_{\ell=1}^{M-1}
    e^{j\pi\left(A\ell^2+\ell(\theta_k-\theta_t)\right)}
    \sum_{m=0}^{M-1-\ell} e^{j2\pi A\ell m}.
\end{align}
Again using the geometric-series identity, the inner sum becomes
\begin{align}
    \sum_{m=0}^{M-1-\ell} e^{j2\pi A\ell m}
    =
    e^{j\pi A\ell(M-1-\ell)}
    \frac{\sin\big(\pi(M-\ell)A\ell\big)}{\sin(\pi A\ell)}.
\end{align}
Hence,
\begin{align}\label{delta_2_simplify_rewrite}
    \zeta_2
    =
    \sum_{\ell=1}^{M-1}
    e^{j\pi \ell\big(A(M-1)+\theta_k-\theta_t\big)}
    \frac{\sin\big(\pi(M-\ell)A\ell\big)}{\sin(\pi A\ell)}.
\end{align}
Finally, substituting \eqref{delta_1_simplify_rewrite} and \eqref{delta_2_simplify_rewrite} into $\nu_k = \frac{\zeta_1+M+\zeta_2}{M^2}$, and using $e^{j\theta} + e^{-j\theta} = 2\cos(\theta)$, we obtain \eqref{nu_A_B}, which proves the lemma.
\section{Proof of Lemma \ref{lemma_1} }\label{appendix_A}
Define $u(p) \triangleq \frac{p}{\sin(\pi Ap)}$. Then, from \eqref{define_f}, $f_p(\theta_t, r_t) = u(p)\sin(\pi(M-p)Ap)$. We first approximate \(u(p)\). From \eqref{define_A_range} $A(\theta_t,r_t) \in \left[0, A_{\max} \right]$. This bound shows that the maximum possible value of $A(\theta_t,r_t)$ decreases as $M$ increases. Therefore, over the entire considered user region, $A(\theta_t,r_t)$ is small for large $M$, which justifies the use of small-angle expansion.

Using $\sin x = x + \mathcal{O}(x^3)$, with \(x=\pi Ap\), we obtain $\sin(\pi Ap) = \pi Ap\left(1+\mathcal{O}\left(A^2p^2\right)\right)$. Substituting this into \(u(p)\) gives $u(p) = \frac{1}{\pi A \left(1+\mathcal{O}\left(A^2p^2\right)\right)}$. Since $p\in [1,M-1]$ and $A \in \left[0, A_{\max} \right]$, we have $A^2p^2 \leq A_{\max}^2 (M-1)^2 = \mathcal{O}(\frac{1}{M})$. Thus, for all \(p\in\{1,2,\ldots,M-1\}\),
\begin{align}
    u(p)
    = \frac{1}{\pi A}\left(1+\mathcal{O}\left(\frac{1}{M}\right)\right).
\end{align}

Hence, for large \(M\), \(u(p)\) is approximately constant over the entire range of interest. Next, notice that the second factor in $f_p(\theta_t, r_t)$ is exactly symmetric under the reflection \(p \mapsto M-p\).
Specifically, $\sin(\pi(M-(M-p))A(M-p)) = \sin(\pi(M-p)Ap)$. Hence, the term $\sin(\pi(M-p)Ap)$ is exactly symmetric about $p = M/2$. Therefore, for large \(M\), $f_p(\theta_t, r_t)$ is approximately even about \(p=M/2\), which proves the lemma.
\section{Proof of Lemma \ref{lemma_2} }\label{appendix_B}
From Lemma \ref{lemma_1}, the function $f_p(\theta_t, r_t)$ is approximately even about $p= \frac{M}{2}$ for large $M$. Therefore, in \eqref{define_derivative},
the summation approximately cancels if $g_p(\theta_k, \theta_t, r_t)$ is odd about \(p=M/2\). Define $C \triangleq A(M-1)+\theta_k-\theta_t$ so that $ g_p(\theta_k, \theta_t, r_t) =\sin(\pi p C)$. For $ g_p(\theta_k, \theta_t, r_t) $ to be odd about $p=M/2$, we require $g_{M-p}(\theta_k, \theta_t, r_t)=-g_p(\theta_k, \theta_t, r_t)$. If $C=\frac{2b}{M}, b\in\mathbb Z$,
then $\pi MC=2b\pi$, and hence $g_{M-p}(\theta_k, \theta_t, r_t) = \sin\left(2b\pi-\pi pC\right) = -g_p(\theta_k, \theta_t, r_t)$. Thus, $g_p(\theta_k, \theta_t, r_t)$ is odd about $p=M/2$ for $\theta_k^{(b)} = \theta_t-A(M-1)+\frac{2b}{M}, b\in\mathbb Z$.
 
To quantify the residual, write
$
h(p)\triangleq \sin(\pi A p(M-p)).
$
Since
$
\frac{p}{\sin(\pi Ap)} = \frac{1}{\pi A} + \mathcal O(Ap^2)$, we have $ f_p(\theta_t,r_t) = \frac{h(p)}{\pi A} + h(p)\mathcal O(Ap^2)$. Since the first term is exactly even, only the residual contributes. Since $|h(p)|\leq1$ and $|g_p(\theta_k,\theta_t,r_t)|\leq1$, from \eqref{define_derivative},
$ \left. \frac{\partial \nu_k}{\partial \theta_k} \right|_{\theta_k=\theta_k^{(b)}} = \mathcal O(AM)$.
Since $A\leq A_{\max}=\mathcal O(M^{-3/2})$, the derivative is $\mathcal O(M^{-1/2})$. Hence, $\theta_k^{(b)}$ is an approximate stationary direction for large $M$.
\section{Proof of Lemma \ref{lemma_3} }\label{appendix_C}
From Lemma \ref{lemma_2}, the approximate stationary far-field directions of the RFF gain are $\theta_k^{(b)} = \theta_t - (M-1)A(\theta_t,r_t) + \frac{2b}{M}, b \in \mbb{N}_{0}$. 
Recall from \eqref{gain_define_critical_points},
\begin{align}
    G_b(A)
    =
    \frac{1}{M}
    +
    \frac{2}{M^2}
    \sum_{p=1}^{M-1}
    \cos\left(\pi p\frac{2b}{M}\right)
    \frac{\sin(\pi(M-p)Ap)}{\sin(\pi Ap)}.
\end{align}
We can rewrite $G_{b}(A)$ by substituting the $b$th candidate into the gain expression in \eqref{Y_k_original} as
\begin{align}
    G_b(A)
    =
    \frac{1}{M^2}
    \left|
    \sum_{n=0}^{M-1}
    e^{j\pi A(n^2-(M-1)n)}
    e^{j2\pi b n/M}
    \right|^2 .
\end{align}

The natural scaling of the curvature parameter is obtained from the quadratic phase term $\pi A n^2$. Since the largest antenna index is of order $M$, the total quadratic phase variation across the array is of order $AM^2$. Therefore, we define the scaled curvature parameter $\tau \triangleq AM^2.$ We now set $A=\tau/M^2$ and $x_n=n/M$. Then $A(n^2-(M-1)n) = \tau(x_n^2-x_n) + \frac{\tau x_n}{M}$. Hence, for fixed $\tau$ and fixed candidate index $q$,
\begin{align}
    e^{j \pi A(n^2 - (M-1)n)} e^{j 2 \pi b n/M} = e^{j \pi (\tau (x_n^2 - x_n) + 2 b x_n)} e^{j \pi \tau x_n/M}.
\end{align}
Since $e^{j \pi \tau x_n/M} = 1 + \mathcal{O}(1/M)$ for all $x_n \in [0,1]$, we get
\begin{align}
    \frac{1}{M}
    \sum_{n=0}^{M-1}
    &e^{j\pi A(n^2-(M-1)n)}
    e^{j2\pi b n/M}
    = \notag \\
    &\frac{1}{M}
    \sum_{n=0}^{M-1}
    e^{j\pi(\tau(x_n^2-x_n)+2 b x_n)}
    +
    \mathcal O\left(\frac{1}{M}\right).
\end{align}
The summation on the right-hand side is a Riemann sum. Therefore,
\begin{align}
    \frac{1}{M}
    \sum_{n=0}^{M-1}
    e^{j\pi A(n^2-(M-1)n)}
    e^{j2\pi b n/M}
    \to
    \int_0^1
    e^{j\pi(\tau(x^2-x)+2 b x)}
    dx
\end{align}
as $M\to\infty$. Therefore, $G_b\left(\frac{\tau}{M^2}\right)
    \to
    \mathcal G_b(\tau)$,
where
\begin{align}\label{define_G_q_alpha}
    \mathcal G_b(\tau)
    \triangleq
    \left|
    \int_0^1
    e^{j\pi(\tau(x^2-x)+2bx)}
    dx
    \right|^2.
\end{align}

The gain expression in \eqref{define_G_q_alpha} has a useful interpretation at $\tau = 0$. For $b=0$, $\mathcal{G}_{0}(0) =1$. For $b \geq 1 $, $\mathcal{G}_{q}(0) = 0$. Thus, when the near-field curvature is negligible, the candidate $q = 0$ is optimal. The integral in \eqref{define_G_q_alpha} can be expressed using the Fresnel integrals. Define $C(u)
    =
    \int_0^u
    \cos\left(\frac{\pi t^2}{2}\right)dt$ and $S(u)
    =
    \int_0^u
    \sin\left(\frac{\pi t^2}{2}\right)dt$. Completing the square in \eqref{define_G_q_alpha} gives
\begin{align}\label{fresnel_gain}
    \mathcal G_b(\tau)
    =
    \frac{1}{2\tau}
    \left[
    \left(C(u_b^+)-C(u_b^-)\right)^2
    +
    \left(S(u_b^+)-S(u_b^-)\right)^2
    \right],
\end{align}
where $u_b^-
    =
    -\sqrt{\frac{\tau}{2}}
    +
    \frac{\sqrt{2}b}{\sqrt{\tau}}$ and $u_b^+
    =
    \sqrt{\frac{\tau}{2}}
    +
    \frac{\sqrt{2}b}{\sqrt{\tau}}$.

For each $b \geq 1$, define $\tau_{\mathrm{peak},b}$ as the first local maximum of $\mathcal{G}_{b}(\tau)$. Numerical calculation of \eqref{fresnel_gain} gives
\begin{align}
\begin{array}{c|c|c}
b
&
\tau_{\mathrm{peak},b}
&
\mathcal G_b(\tau_{\mathrm{peak},b})
\\
\hline
1 & 5.28818 & 0.27533 \\
2 & 8.34186 & 0.17585 \\
3 & 11.13386 & 0.13140 \\
4 & 13.79617 & 0.10566 \\
5 & 16.37723 & 0.08871
\end{array}
\end{align}
Thus, the first local maximum shifts to larger values of $\tau$ as $b$ increases. This shows that higher order candidates become relevant only for larger values of $A(\theta_t,r_t)$. This makes intuitive sense because a larger $A(\theta_t,r_t)$ corresponds to a higher mismatch and a significant beam-split gain pattern.

We now define the transition points. Let $\beta_b$, with $\beta_0 = 0$, denote the scaled transition values corresponding to the sequence $\{ A_{\textrm{th},b} \}$, i.e., $A_{\textrm{th},b} = \frac{\beta_b}{M^2}$. Equivalently, $\beta_b$ is the value of $\tau$ that satisfies $\mathcal{G}_{b-1}(\beta_b) = \mathcal{G}_{b}(\beta_b).$
Numerically evaluating $\{ A_{\textrm{th},b} \}$ from \eqref{define_B}, we get
\begin{align}
\beta_1=4.72317,\quad
\beta_2=7.78690,\quad
\beta_3=10.58573.
\end{align}

We first consider the interval $\beta_1 \leq \tau \leq \beta_2$. The relevant scaled values are
\begin{align}
    \beta_1=4.72317,
    \quad \tau_{\mathrm{peak},1}=5.28818,
    \quad
    \beta_2=7.78690.
\end{align}

The gains using \eqref{fresnel_gain} at these values are
\begin{align}
\begin{array}{c|c|c|c}
b
&
\mathcal G_b(\beta_1)
&
\mathcal G_b(\tau_{\mathrm{peak},1})
&
\mathcal G_b(\beta_2)
\\
\hline
0 & 0.26825 & 0.19072 & 0.08540 \\
1 & 0.26825 & 0.27533 & 0.17253 \\
2 & 0.07458 & 0.09642 & 0.17253 \\
3 & 0.01611 & 0.02308 & 0.07466 \\
4 & 0.00414 & 0.00598 & 0.02405 \\
5 & 0.00140 & 0.00197 & 0.00773
\end{array}
\end{align}
At the left endpoint $\beta_1$, the candidates $b=0$ and $b=1$ tie. At the right endpoint $\beta_2$, the candidates $b=1$ and $b=2$ tie. Over this interval, the previous candidate $b=0$ is decreasing, while the higher-order candidates $b \geq 2$ have their first local maxima outside the interval and remain below $b = 1$, with $b = 2$ only tying at $\beta_2$. Therefore, $b = 1$ gives the largest gain on $\beta_1 \leq \tau \leq \beta_2$.

Similarly, consider the interval $\beta_2 \leq \tau \leq \beta_3$. The relevant scaled values are
\begin{align}
    \beta_2=7.78690,
    \quad
    \tau_{\mathrm{peak},2}=8.34186,
    \quad
    \beta_3=10.58573.
\end{align}
The gains using \eqref{fresnel_gain} at these values are
\begin{align}
\begin{array}{c|c|c|c}
b
&
\mathcal G_b(\beta_2)
&
\mathcal G_b(\tau_{\mathrm{peak},2})
&
\mathcal G_b(\beta_3)
\\
\hline
0 & 0.08540 & 0.09656 & 0.13203 \\
1 & 0.17253 & 0.13865 & 0.06056 \\
2 & 0.17253 & 0.17585 & 0.12940 \\
3 & 0.07466 & 0.08866 & 0.12940 \\
4 & 0.02405 & 0.03099 & 0.06838 \\
5 & 0.00773 & 0.01019 & 0.02736
\end{array}
\end{align}
At the left endpoint $\beta_2$, the candidates $b=1$ and $b=2$ tie. At the right endpoint $\beta_3$, the candidates $b=2$ and $b=3$ tie. Over this interval, the previous candidate $b=1$ is decreasing, while the higher-order candidates $b \geq 3$ have their first local maxima outside the interval and remain below $b=2$, with $b=3$ only tying at $\beta_3$. Therefore, $b=2$ gives the largest gain on $\beta_2 \leq \tau \leq \beta_3$, except on the sub-interval $(10.50374,10.81127)$ where $\mathcal{G}_0(\tau)$ dominates. This is the exception noted in Remark $\ref{remark_execption}$.

The same pattern is observed for subsequent intervals by numerically evaluating the gains. Hence, for large $M$, if $\beta_{b^{\star}} \leq AM^2 \leq \beta_{b^{\star}+1}$,
or equivalently $A_{\textrm{th},b^{\star}} \leq A \leq A_{\textrm{th},b^{\star}+1}$, then the maximizing far-field direction is $\theta_k^{(b^{\star})}$, up to the exception noted above. Substituting the direction $\theta_k^{(b^{\star})}$ into \eqref{nu_A_B} gives \eqref{max_ff_cb}. This proves the large-$M$ approximation.
\section{Proof of Lemma \ref{lemma_nf}}\label{appendix_nf}
Define the phase difference vector $\boldsymbol{\xi} \triangleq [\xi_0,\ldots,\xi_{M-1}]^{\mathsf T}$, where $\xi_n = \frac{2\pi}{\lambda} \left(r^{(n)}(\theta_k,r_k) - r^{(n)}(\theta_t,r_t) \right)$. Then, the beamforming gain of the $k$-th RNF codeword is $\nu_k = \frac{1}{M^2} \left| \sum_{n=0}^{M-1} e^{-j\xi_n} \right|^2$. Let $\Delta\theta=\theta_k-\theta_t$, $\Delta r=r_k-r_t$, and $ \mbf{o} = \begin{bmatrix} \Delta\theta \, \Delta r \end{bmatrix}^{\mathsf{T}}$. For small offsets around $(\theta_t,r_t)$, a first-order Taylor expansion of the phase map gives $\boldsymbol{\xi} = \mathbf J \mbf{o} + \mathcal O(\| \mbf{o} \|^2)$,
where $\mathbf J\in\mathbb R^{M\times 2}$ is the Jacobian of the phase vector with respect to $(\theta,r)$, evaluated at $(\theta_t,r_t)$.

Using the second-order Taylor expansion $e^{-j\xi_n} = 1-j\xi_n-\frac{\xi_n^2}{2} + \mathcal O(\xi_n^3)$, we obtain
\begin{align}
    \sum_{n=0}^{M-1}e^{-j\xi_n}
    =
    M
    -
    j\sum_{n=0}^{M-1}\xi_n
    -
    \frac{1}{2}\sum_{n=0}^{M-1}\xi_n^2
    +
    \mathcal O(\|\boldsymbol{\xi}\|^3).
\end{align}
Therefore,
\begin{align}
    \left|
    \sum_{n=0}^{M-1}e^{-j\xi_n}
    \right|^2
    =
    M^2
    -
    M\sum_{n=0}^{M-1}\xi_n^2
    +
    \left(\sum_{n=0}^{M-1}\xi_n\right)^2
    +
    \mathcal O(\|\boldsymbol{\xi}\|^3).
\end{align}
Dividing by $M^2$ gives
\begin{align}
    \nu_k
    =
    1
    -
    \frac{1}{M}\sum_{n=0}^{M-1}\xi_n^2
    +
    \frac{1}{M^2}
    \left(\sum_{n=0}^{M-1}\xi_n\right)^2
    +
    \mathcal O(\|\boldsymbol{\xi}\|^3).
\end{align}
Now define $\mathbf P \triangleq \mathbf I_M-\frac{1}{M}\mathbf 1_M\mathbf 1_M^{\mathsf T}$, where $\mathbf I_M$ is the $M \times M$ identity matrix and $\mathbf 1_M$ is the $M \times 1$ all-ones vector. Since
\begin{align}
    \|\mathbf P\boldsymbol{\xi}\|^2
    =
    \sum_{n=0}^{M-1}\xi_n^2
    -
    \frac{1}{M}
    \left(\sum_{n=0}^{M-1}\xi_n\right)^2,
\end{align}
we can write
\begin{align}
    \nu_k
    =
    1
    -
    \frac{1}{M}
    \|\mathbf P\boldsymbol{\xi}\|^2
    +
    \mathcal O(\|\boldsymbol{\xi}\|^3).
\end{align}
Using $\boldsymbol{\xi} = \mathbf J \mbf{o} + \mathcal O(\| \mbf{o} \|^2)$ inside $\lVert \mbf{P} \boldsymbol{\xi} \rVert^2$ gives
\begin{align}
    \|\mathbf P\boldsymbol{\xi}\|^2
    =
    \mbf{o} ^{\mathsf T}
    \mathbf J^{\mathsf T}\mathbf P\mathbf J
    \mbf{o}
    +
    \mathcal O(\|\mbf{o}\|^3).
\end{align}
Thus,
\begin{align}
    \nu_k
    =
    1
    -
    \frac{1}{M}
    \mbf{o}^{\mathsf T}
    \mathbf J^{\mathsf T}\mathbf P\mathbf J
    \mbf{o}
    +
    \mathcal O(\|\mbf{o}\|^3).
\end{align}
Let $\mathbf Q \triangleq \mathbf J^{\mathsf T}\mathbf P\mathbf J$. Thus, the high-gain event
$\nu_k\ge 1-\epsilon$ is, up to higher-order terms, equivalent to $\mbf{o}^{\mathsf T} \mathbf Q \mbf{o} \le M\epsilon$. This inequality defines an ellipse in the $(\Delta\theta,\Delta r)$ plane. Its semi-axis lengths scale as $\sqrt{\epsilon}$, so its area scales linearly with $\epsilon$. The cubic remainder perturbs the boundary by order $\epsilon^{3/2}$, and hence
\begin{align}\label{cdf_ap1}
    \Pr\{\nu_k\ge 1-\epsilon\}
    =
    G_M\epsilon
    +
    \mathcal O(\epsilon^{3/2}),
\end{align}
where $G_M>0$ depends on the local sampling density of $(\theta_k,r_k)$ and on the local geometry through $\mathbf Q$. Using \eqref{cdf_ap1} together with $ 1-F_\nu(1-\epsilon) = \Pr\{\nu\geq 1-\epsilon\}$, and then setting $\epsilon=1-\gamma$ proves the lemma.
\end{document}